\documentclass[11pt]{article}%
\usepackage{amsfonts}
\usepackage{amsmath,bbm,dsfont,mathrsfs,mathtools,appendix}
\usepackage{amssymb}
\usepackage{todonotes}
\usepackage{graphicx}
\usepackage{fullpage}%
\usepackage[colorlinks=true,linkcolor=blue,citecolor=red,plainpages=false,pdfpagelabels]%
{hyperref}%
\providecommand{\U}[1]{\protect\rule{.1in}{.1in}}

\newtheorem{theorem}{Theorem}

\newtheorem{lemma}[theorem]{Lemma}

\newtheorem{proposition}[theorem]{Proposition}
\newtheorem{remark}[theorem]{Remark}

\newenvironment{proof}[1][Proof]{\noindent\textbf{#1.} }{\ \rule{0.5em}{0.5em}}

\newcommand{\cF}{{\mathcal{F}}}

\newcommand{\cB}{{\mathcal{B}}}

\newcommand{\cH}{\mathcal{H}}
\newcommand{\R}{\mathbbm{R}}

\newcommand{\N}{\mathbbm{N}}

\newcommand{\nn}{\nonumber}
\newcommand{\bP}{\mathbbm{P}}
\newcommand{\bE}{\mathbbm{E}}

\newcommand{\supp}{\operatorname{supp}}

\newcommand{\X}{\mathcal{X}}

\newcommand{\cP}{\mathcal{P}}

\newcommand{\Z}{\mathcal{Z}}

\newcommand{\cS}{\mathcal{S}}

\newcommand{\cE}{\mathcal{E}}

\newcommand{\cX}{{\cal X}}

\newcommand{\cM}{\mathcal{M}}
\newcommand{\cN}{\mathcal{N}}

\newcommand{\1}{\mathbbm{1}}
\newcommand{\cD}{\mathcal{D}}

\def\>{{\rangle}}
\def\<{{\langle}}
\newcommand{\be}{\begin{equation}}
\newcommand{\ee}{\end{equation}}
\newcommand{\bea}{\begin{eqnarray}}
\newcommand{\eea}{\end{eqnarray}}

\newcommand{\eps}{\varepsilon}
\newcommand{\teps}{{\tilde{\varepsilon}}}

\newcommand{\ket}[1]{|#1\rangle} 
\newcommand{\bra}[1]{\langle#1|} 
\newcommand{\kb}[1]{|#1\rangle\!\langle#1|} 

\newcommand{\Tr}{\mathrm{Tr}}

\newcommand{\comment}[1]{}
\newcommand{\red}[1]{{\color{red} #1}}

\newcommand{\stc}{\rm{sc}}

\DeclarePairedDelimiter\floor{\lfloor}{\rfloor}

\newcommand{\tr}{{\rm Tr}}

\usepackage{color}
\definecolor{colorthree}{rgb}{0.01,0.51,0.93}

\title{Total insecurity of communication via strong converse for quantum privacy amplification}

\allowdisplaybreaks

\begin{document}
{\author{ Robert Salzmann\thanks{University of Cambridge, Department of Applied
 and Theoretical Physics, Wilberforce Road, Cambridge CB3 0WA,
United Kingdom}
 \and Nilanjana Datta\footnotemark[1]}}
\maketitle

\begin{abstract}
Quantum privacy amplification is a central task in quantum cryptography.
Given shared randomness, which is initially correlated  with a quantum system held by an eavesdropper, the goal is to extract uniform randomness which is decoupled from the latter. The optimal rate for this task is known to satisfy the strong converse property and we provide a lower bound on the corresponding strong converse exponent. In the strong converse region, the distance of the final state of the protocol from the desired decoupled state converges exponentially fast to its maximal value, in the asymptotic limit.  We show that this necessarily leads to totally insecure communication by establishing that the eavesdropper can infer any sent messages with certainty, when given very limited extra information. In fact, we prove that in the strong converse region, the eavesdropper has an exponential advantage in inferring the sent message correctly, compared to the
achievability region. Additionally we establish the following technical result, which is central to our proofs, and is of independent interest: the smoothing parameter for the smoothed max-relative entropy satisfies the strong converse property. 
\end{abstract}

\section{Introduction}

Finding optimal rates of information-theoretic tasks, such as compression of information for efficient storage, transmission of information through noisy quantum
channels, privacy amplification and entanglement manipulation, are of fundamental interest in quantum information theory. 
Depending on the specific task at hand, the optimal rate is either an optimal gain, quantifying the maximum
rate at which a desired target resource\footnote{A quantum information source, a noisy quantum channel, shared- entanglement or randomness, are examples of resources in quantum information theory} can be produced in the process, or an optimal cost, quantifying the minimum rate 
at which an available resource is consumed in accomplishing the task. Any rate which lies below (resp.~above)
the optimal gain (resp.~cost) is said to be an {\em{achievable rate}}. This is because, for any such rate,
there is a corresponding protocol for accomplishing the task successfully,
i.e., such that the error, $\eps_n$, incurred in the protocol for $n$ successive uses of the underlying resource vanishes in
the so-called asymptotic limit ($n \to \infty$). In contrast, for protocols
with non-achievable rates, the error does not vanish asymptotically. The optimal rate of an
information-theoretic task is said to satisfy the {\em{strong converse property}} if any sequence of
protocols with a non-achievable rate fails with certainty in the asymptotic limit. That is, the incurred error, $\eps_n$, is not only bounded 
away from zero but necessarily converges to one in the asymptotic limit. Often this convergence can be shown to be exponential
in $n$, i.e. \begin{align*}\eps_n \sim  1 - 2^{-rn}\end{align*}
for some positive constant $r$. In this case the smallest such constant is called the {\em{strong converse exponent}} of the task.

The above mentioned optimal rates are typically evaluated in the so-called asymptotic, memoryless setting, namely, a setting in which one assumes that $(i)$ there is no correlation between successive uses of the underlying resource and $(ii)$ the latter is available for arbitrarily many uses; moreover, one demands that the error vanishes in the asymptotic limit. The optimal rates are given by entropic quantities stemming from the quantum relative entropy~\cite{Umegaki_ConditionalExpectation_1962}, which hence serves as a parent for these quantities. In contrast, in the so-called one-shot or finite blocklength setting, one considers just a single use (resp.~a finite number of uses) of the resource, and hence it is unrealistic to demand that the error vanishes. So one allows for a small but non-zero error (say, $\eps \in (0,1)$). The analogous entropic quantities in this case are given by {\em{smoothed}} entropic functions, stemming from corresponding parent quantities which are {\em{smoothed}} generalised divergences, with the smoothing parameter being given by the error threshold $\eps$. Examples of such divergences are the smoothed max- and min-relative entropies\footnote{This is also known as the hypothesis testing relative entropy \cite{WangRenner_OneShotClassCapacity_2012}.}. For any $\eps \in (0,1)$, the smoothed divergence between two tensor power states (e.g. $\rho^{\otimes n}$ and $\sigma^{\otimes n}$) has been shown to reduce to the quantum relative entropy, $D(\rho\|\sigma)$, in the limit $n \to \infty$. This important property is called the quantum asymptotic equipartition property (AEP). This property and the fact that the smoothing parameter corresponds to an error, leads to a natural definition of a strong converse exponent of a smoothed generalized divergence, without reference to any information-theoretic task. For example, let us consider the smoothed max-relative entropy, $D_{\max}^\eps(\rho^{\otimes n}\|\sigma^{\otimes n} )$, of a pair of tensor power (or i.i.d.) states $\rho^{\otimes n}$ and $\sigma^{\otimes n}$. It is known that if this divergence is constrained to be greater than $nr$ for any fixed $r > D(\rho\|\sigma)$, then the smoothing parameter is known to vanishes asymptotically as $n \to \infty$. Recently, the exact exponent with which it vanishes, was evaluated by Li, Yao and Hayashi~\cite{LiYao_ExponentSmoothMaxEntr_2021}\footnote{This was done in the case in which the smoothed divergence is defined in terms of the purified distance.}. In this paper, we show that in contrast, if one constrains $D_{\max}^\eps(\rho^{\otimes n}\|\sigma^{\otimes n} )$ to be less $nr$ for some $r < D(\rho\|\sigma)$, then the smoothing  parameter converges to one exponentially. We call the exact exponent of this convergence the {\em{strong converse exponent}} of the smoothed max-relative entropy, and provide a lower a bound on it which is tight if $\rho$ and $\sigma$ commute.

 In this paper we focus on the task of quantum privacy amplification (also known as randomness extraction) which is of central importance in quantum cryptography. In it two distant parties (say, Alice and Bob) initially share some randomness, given by a random variable $X$, which is only partially secure, in the sense that an adversary (say, Eve) holds a system, $E$, which is correlated with $X$. The aim of Alice and Bob is to distil uniform randomness (or shared secret key) from $X$, using one-way communication, such that the resulting key is uncorrelated with Eve's system, and is hence secure. This key can then be employed for secure communication between Alice and Bob. The case in which Eve's system $E$ is a classical random variable, was studied in \cite{BennettBrassardRobert_PrivAmpl_1988,ImpagliazzoLevinLuby_PseudoRandom(PrivacyAmpl)_1989,BennettBrassardCrepeauMaurer_GenPrivAmpl_1995}. The case in which Eve is a quantum adversary, that is, her system $E$ is a quantum system, was studied in \cite{KoenigMaurerRenner_PowerQuantumMemory_2005,Renner_PhDThesis_2005,RennerKoenig_UniversiallyCompPrivAmpl_2005}. The optimal rate of privacy amplification (or secret key distillation) in both these cases, in the asymptotic memoryless setting, was shown to be given by the conditional entropy $H(X|E)$. Any attempt to distil uniform randomness at rates greater than this conditional entropy, fails with certainty, with the error incurred in the protocol converging exponentially to one. In this paper, we evaluate the strong converse exponent for this task, i.e.~the speed of this exponential convergence, in the case in which Eve is a quantum adversary.  

As mentioned above, a privacy amplification protocol is considered to be successful if the final shared randomness between Alice and Bob is close to uniform, and Eve has negligible information about it. Mathematically, in the asymptotic memorlyless setting, one requires that the distance between the final state of the protocol\footnote{This is a classical-quantum state with the classical part corresponding to the final shared randomness and the quantum part is with Eve.} and a decoupled state, in which the classical system is completely mixed (and hence corresponds to uniform randomness) and is uncorrelated with Eve's system, vanishes in the asymptotic limit. It is clear that this ensures that the generated key is secure and hence can be used subsequently for secure communication between Alice and Bob. In contrast, in the strong converse region, this distance converges exponentially to one. Hence one would expect that if Alice employs keys which are generated at an asymptotic rate greater than  $H(X|E)$ to send messages to Bob, then the communication is {\em{totally insecure}}. In this paper we additionally provide a precise mathematical meaning to the notion of {\em{total insecurity}} by establishing that with very little extra information Eve can infer Alice's messages with certainty in the asymptotic limit. Thus we prove a strong converse theorem for secure communication.

\subsection{The setup of quantum privacy amplification}
Let the random variable $X \sim p_x$, $x \in \mathcal{X}, $ denote the common shared randomness between two distant parties (Alice and Bob) at the start of a privacy amplification process. It is initially correlated with a quantum system, $E$, held by an eavesdropper (Eve). The initial state of the process is hence represented by the classical-quantum (c-q) state
\begin{align}
\label{eq:cqstate}
\rho_{XE} = \sum_{x\in\X}p_x \kb{x}\otimes\rho_E^{x}.
\end{align}
Alice and Bob apply a (hash) function $f: \cX \to \Z$, with $|\Z| < |\cX|$, with the aim of extracting a random variable $Z$ which is uniformly distributed and independent of the state of $E$. 
The goal is to $(i)$ maximize the {\em{size of the extracted randomness}}, given by $|\Z|$, and $(ii)$ minimize the \emph{conversion error}, i.e., the distance between the resulting state, $ \rho^f_{ZE}$, of the protocol and the desired decoupled state $\frac{\1_\Z}{|\Z|}\otimes \rho_E$. Here we have denoted Eve's average state by $\rho_{E} = \sum_{x} p_x\rho_E^x$. 

We denote by $l^\eps(X|E)$ the largest key length, $\log|\Z|$, for which the conversion error can be made less than or equal to a fixed $\eps\ge 0$. In the so-called asymptotic i.i.d.~setting, $n$ copies of the c-q state \eqref{eq:cqstate} are available and one is interested in the asymptotic behaviour of the key length $l^\eps(X^n|E^n)$ as $n$ goes to infinity. It is well-known~\cite{Renner_PhDThesis_2005,KoenigRenner_SamplingMinEntropy_2011} that the optimal key rate is given by the conditional entropy of the state $\rho_{XE}$, i.e.,
\begin{align}
\label{eq:AEPkeyrate}
\lim_{n\to\infty}\frac{l^{\eps}(X^n|E^n)}{n} = H(X|E),
\end{align}
for all $0<\eps<1$.  
Hence, $H(X|E)$ has the operational interpretation in the context of privacy amplification of being the optimal rate at which secret keys can be generated which are uniformly distributed and independent of Eve's system in the asymptotic limit.

From \eqref{eq:AEPkeyrate} we can already infer the strong converse property of privacy amplification: If Alice and Bob apply a hash function $f:\X^n\to\Z_n$ with $|\Z_n| = \left \lfloor{2^{nR}}\right \rfloor$ and  the {\emph{privacy amplification rate}, $R$, satisfies $R> H(X|E)$, then the conversion error necessarily tends to one in the asymptotic limit.

\section{Main results}

Below we summarize the main results of our paper.

\begin{itemize}
\item
\textbf{Strong converse exponent of privacy amplification:} We prove that the conversion error for a privacy amplification scheme with rate $R>H(X|E)$
converges to $1$ exponentially fast, in the asymptotic limit. Moreover, we provide a lower bound on the corresponding strong converse exponent\footnote{Note that bounds on the strong converse exponent for privacy amplification were also obtained in~\cite{LeditzkyWildeDatta_StrongConverseRenyiEnt_2016} under a different choice of the figure of merit. Moreover, note that in a concurrent and independent work \cite{ShenGaoCheng_StrongConvPrivAmpl_2022}, a comparable bound for the average conversion error (over strongly 2-universal hash functions) in trace distance is obtained. See Section~\ref{sec:PrivacyAmplification} for more details.}.
In particular, if $\Z_n$ is the output classical register with $|\Z_n| = \left\lfloor{}2^{ nR}\right\rfloor$, and $\Delta_P(\X^n\to\Z_n)$ is the {\emph{minimal conversion error}}, i.e.~smallest attainable (purified) distance of the final state from the desired decoupled state, we show that the strong converse rate, $sc_P(R)$, satisfies:
\begin{align}
\label{eq:StrongConverseRate}
sc_P(R) &\coloneqq\liminf_{n\to\infty}\frac{-\log\left(1-\Delta_P(\X^n \to \Z_n)\right)}{n} \ge \sup_{0<\alpha<1}(1-\alpha)\left(R -H_\alpha(X|E)\right),
\end{align}
with  $H_\alpha(X|E) := \frac{1}{1-\alpha}\log\Tr(\rho_{XE}^\alpha(\1_X\otimes\rho_E)^{1-\alpha})$ being the $\alpha$-conditional R\'enyi entropy.
\item
\textbf{Strong converse of secure communication:} Most of the literature on privacy amplification was concerned with establishing security proofs showing that if the distance to the decoupled state can be made small using a suitable hash function, the probability with which Eve can guess the generated key is also small. On the other hand, the implications of large distance to decoupled state, as is the case in the strong converse region discussed above, are less obvious.

In order to clarify in what sense large conversion error implies insecurity in the communication with the generated key, we consider Alice encrypting a message using the key, and evaluate the extra side information Eve needs to have in order to guess the message with certainty. \comment{Here, Alice is using a one-time pad encryption scheme in which she encodes her message $m\in\Z_n$ with a secret key $z\in\Z_n$ by \red{component wise addition} $m\oplus z$ and then publicly sends the encrypted message.} We consider the scenario in which Alice chooses her message from some subset $\cM_n\subset\Z_n$. Here the set $\cM_n$ is known by Eve and its size determines Eve's additional side information: Eve's uncertainty about Alice's message increases with the cardinality $|\cM_n|$.

We show below that whenever the minimal conversion error is strictly smaller than 1, Eve needs to have very strong additional side information about the set of messages, i.e. $|\cM_n|$ needs to be finite (uniformly in $n$) in order for Eve to be able to guess Alice's message with certainty. 

In contrast, in the strong converse region, in which the conversion error approaches 1 exponentially fast, Eve can infer the message even for sets of messages $\cM_n$ which grow exponentially in $n$.
In fact, we show that almost all\footnote{Here 'almost all' means that the proportion of subsets $\cM_n$ with size constraint $|\cM_n|\ll 2^{nsc_P(R)/2}$ for which the statement holds approaches 1 as goes to $n\to\infty.$ } sets with size $|\cM_n|\ll 2^{nsc_P(R)/2}$ are such that if Alice picks a message $M$ out of $\cM_n$, the probability with which Eve is able to guess $M$ correctly approaches 1 as $n\to\infty$. Moreover, this convergence can be shown to happen exponentially fast.

These results imply that any communication from Alice to Bob, using keys generated at a rate $R>H(X|E)$, is {\em{totally insecure}}, in the following sense: Eve, with very limited additional information about the set $\cM_n \subset \cX_n$ from which Alice chooses her messages, is able to perfectly infer her messages in the asymptotic limit.
\item
\textbf{Strong converse exponent for smoothed max-relative entropy:} As a technical result which is used in the proof of \eqref{eq:StrongConverseRate}, we show that the smoothing parameter in the smoothed max-relative entropy $D^{\eps}_{\max}(\rho\|\sigma)$ of two states $\rho,\sigma$ necessarily converges to 1 in the asymptotic limit, if one demands that the quotient $D^{\eps}_{\max}(\rho^{\otimes n}\|\sigma^{\otimes n})/n$ is less than the relative entropy $D(\rho\|\sigma)$. We call this the \emph{strong converse property of the smoothed max-relative entropy}.\comment{ From an operational point of view this can be seen as the strong converse in the dilution task of the resource theory of asymmetric distinguishability \cite{WangWilde_RTAD_2019}.} Moreover, we establish that this convergences happens exponentially fast and provide lower bounds on the corresponding strong converse exponent which can be shown to be tight in the case of commuting states.
This result can be seen as the corresponding converse statement of the recent result \cite[Theorem 6]{LiYao_ExponentSmoothMaxEntr_2021} of Li, Yao and Hayashi. For the proof of this result we employ the quantum Hoeffding bound \cite{Hayashi_Hoeffding(ErrorExponent)_2007,Nagaoka_ConverseHoeffding_2006,AudenaertNussbaum_AsymptoticErrorRates_2008} from binary quantum hypothesis testing.
\end{itemize}

\section{Preliminaries} 
\label{sec:Preliminaries}
\subsection{Purified distance and generalised trace distance}
Let $\cH$ be a finite dimensional Hilbert space. We will denote the set of positive semi-definite operators on $\cH$ by $\cP(\cH)$ and moreover by $\cS_{\le}(\cH) = \{\rho\in \cP(\cH)|\, \Tr(\rho)\le 1\}$ and $\cS(\cH) = \{\rho\in \cP(\cH)|\, \Tr(\rho)= 1\}$ the set of sub-normalised and normalised states respectively. For  $\rho,\sigma\in\cS_\le(\cH)$ the generalised trace distance is defined as
$$\cD(\rho,\sigma) := \frac{1}{2}\|\rho-\sigma\|_1 + \frac{1}{2}|\Tr(\rho-\sigma)|.$$  Note that it can be expressed as
\begin{align}
	\label{eq:GenTrace}
	\cD(\rho,\sigma) = \max_{0\le\Lambda\le\1}|\Tr\left(\Lambda(\rho-\sigma)\right)|.
	\end{align}
Moreover, the purified distance \cite{TomamichelColbeckRenner_DualitySmoothMinMax_2010} is defined as 
\begin{align*}
P(\rho,\sigma) := \sqrt{1-F(\rho,\sigma)^2},
\end{align*}
where 
\begin{align*}
 F(\rho,\sigma) := \Tr\left(\sqrt{\sqrt{\rho}\sigma\sqrt{\rho}}\right) + \sqrt{(1-\Tr(\rho))(1-\Tr(\sigma))},
\end{align*}
is the generalised fidelity.
Both the generalised trace distance as well as the purified distance are metrics on the set $\cS_{\le}(\cH)$. From the Fuchs-van de Graaf inequalities \cite{FuchsvandeGraaf_Inequality_1999}, the relations between the generalised trace distance and the purified distance follow \cite[Lemma 6]{TomamichelColbeckRenner_DualitySmoothMinMax_2010}:
\begin{align}
\label{eq:GenPuriBounds}
\cD(\rho,\sigma) \le P(\rho,\sigma) \le \sqrt{2\cD(\rho,\sigma)}.
\end{align}
\subsection{Quantum relative entropies}
The quantum relative entropy \cite{Umegaki_ConditionalExpectation_1962} of $\rho\in\cS_{\le}(\cH) $ with respect to $\sigma\in\cP(\cH)$ is defined as
\begin{align*}
D(\rho\|\sigma) := \Tr\left(\rho\left(\log\rho-\log\sigma\right)\right).
\end{align*}
Moreover, for $\alpha \in(0,1)\cup (1,\infty)$ the \emph{Petz R\'enyi relative entropy} \cite{Petz_Quasi-entropies_1986} of order $\alpha$ is defined as
\begin{align*}
D_\alpha(\rho\|\sigma) := \frac{1}{\alpha-1}\log\Tr\left(\rho^\alpha\sigma^{1-\alpha}\right)
\end{align*}
and the \emph{sandwiched R\'enyi relative entropy} \cite{Muller-lennert_QuantumSandwiched_2013,Wilde_StrongConverseSandwiched_2014} of order $\alpha$ as
\begin{align*}
\widetilde D_\alpha(\rho\|\sigma) \coloneqq \frac{1}{\alpha-1}\log\Tr\left(\sigma^{\frac{1-\alpha}{2}}\rho\sigma^{\frac{1-\alpha}{2}}\right).
\end{align*}
The \emph{max-relative entropy} \cite{Datta_MinMaxRelativeEntr_2009} is defined as
\begin{align*}
D_{\max}(\rho\|\sigma) := \inf\left\{\lambda\,| \,\rho\le 2^\lambda \sigma\right\}.
\end{align*}
These relative entropies fulfill the following relations
\begin{align*}
\widetilde D_\alpha(\rho\|\sigma) \le  D_\alpha(\rho\|\sigma),\quad \lim_{\alpha\to 1 }\widetilde D_\alpha(\rho\|\sigma)= \lim_{\alpha\to 1 }  D_\alpha(\rho\|\sigma) = D(\rho\|\sigma),\quad  \lim_{\alpha\to \infty }\widetilde D_\alpha(\rho\|\sigma) = D_{\max}(\rho,\sigma).
\end{align*}
Moreover, both the Petz- and sandwiched R\'enyi relative entropies are monotonically increasing in the parameter $\alpha$. 

For $\eps\ge 0$ and $d$ denoting a metric on $\cP(\cH)$ we define the \emph{smoothed max-relative entropy} \cite{Datta_MinMaxRelativeEntr_2009} to be
\begin{align*}
D^{\eps,d}_{\max}(\rho\|\sigma) := \inf_{\widetilde\rho\in\cB^{d}_{\eps}(\rho)} D_{\max}(\widetilde\rho\|\sigma).
\end{align*}
Here, we have denoted the ball of sub-normalised states with radius $\eps$ around $\rho$ by $\cB^{d}_\eps(\rho) = \{\widetilde\rho\in\cS_{\le}(\cH) |\, d(\widetilde\rho,\rho) \le\eps\}.$ We will consider the metric $d$ to be either the generalised trace distance or the purified distance, i.e. $d\in\{\cD,P\}.$

Using the above definitions of the various relative entropies, we can define the the corresponding conditional entropies for a bipartite state $\rho_{AB} \in \cS(\cH_A\otimes\cH_B)$ as follows
\begin{align*}
H(A|B) := -D(\rho_{AB}\| \1_A\otimes \rho_B),\quad &H_\alpha(A|B) := -D_{\alpha}(\rho_{AB}\| \1_A\otimes \rho_B),\\ \widetilde H_\alpha(A|B) := -\widetilde D_{\alpha}(\rho_{AB}\| \1_A\otimes \rho_B),\quad
&H^{\eps,d}_{\min}(A|B) := -D_{\max}^{\eps,d}(\rho_{AB}\| \1_A\otimes \rho_B).
\end{align*}
\subsection{The Hoeffding bound of binary quantum hypothesis testing}
\label{sec:Hoeffding}

In binary hypothesis testing the task is to discriminate between two quantum states $\rho$ and $\sigma$ using a measurement, i.e. a POVM $\{\Lambda,\1-\Lambda\}$ with $0\le\Lambda\le\1.$ Here, we interpret $\Tr((\1-\Lambda)\rho)$ as the \emph{type I error probability}, i.e. the probability that $\rho$ was wrongfully infered to be $\sigma$, and $\Tr(\Lambda\sigma)$ as the \emph{type II error probability}, i.e. the probability that $\sigma$ was wrongfully infered to be $\rho$. 

In the assymptotic i.i.d. setting with $n$ copies of the states available, assuming $\rho\neq \sigma$, both type I and type II error probabilities can be made exponentially small in $n$ choosing a suitable measurement.
Here, for $s\ge0$ the optimal optimal \emph{type I error exponent} under the constraint that the \emph{type II error exponent} is greater or equal to $s$ is given by
	\begin{align}
	\label{eq:HoeffdingDef}
	B(s|\rho\|\sigma) &\coloneqq \sup_{\substack{(\Lambda_n)_{n\in\N}\\0\le\Lambda_n\le\1}}\left\{\liminf_{n\to\infty}\frac{-\log\Big(\Tr\left((\1-\Lambda_n)\rho^{\otimes n}\right)\Big)}{n}\Bigg|\liminf_{n\to\infty}\frac{-\log\Big(\Tr\left(\Lambda_n\sigma^{\otimes n}\right)\Big)}{n} \ge s,\right\}.
	\end{align}
This quantity is called the \emph{quantum Hoeffding bound} and it has been proven for $s>0$ that $B(s|\rho\|\sigma)$ can be expressed using the Petz R\'enyi relative entropy \cite{Hayashi_Hoeffding(ErrorExponent)_2007,Nagaoka_ConverseHoeffding_2006,AudenaertNussbaum_AsymptoticErrorRates_2008}
	\begin{align}
	\label{eq:HoeffResult}
	B(s|\rho\|\sigma)&= \sup_{0\le\alpha\le 1} \frac{\alpha-1}{\alpha}\left(s - D_\alpha(\rho\|\sigma)\right).
	\end{align}
	For the proof of our Theorem~\ref{thm:StrongConverse} we will for $r\in\R$ fixed use the quantity
	\begin{align*}
	s_r := \sup_{0\le\alpha\le 1}\big(\alpha r -(\alpha-1)D_{\alpha}(\rho\|\sigma)\big),
	\end{align*}
	which can be easily seen to be the unique solution to the equation $B(s|\rho\|\sigma) = s - r.$
    It can be shown that the so-called the Neyman-Pearson test is optimal in \eqref{eq:HoeffdingDef}, which is sumarised in the following lemma  (see e.g. \cite[Theorem 1.4] {AudenaertMosonyi_QStateDiscriminationBounds_2012}\footnote{Note that in \cite{AudenaertMosonyi_QStateDiscriminationBounds_2012} the divergence $H_s(\rho\|\sigma)$ satisfies in our convention $H_s(\rho\|\sigma)=B(s|\sigma\|\rho).$} )
	\begin{lemma} [Asymptotics of Neyman-Pearson tests]\label{lem:NeymanPearson}
	Let $r\in\R$ and $s_r$ be such that $= B(s_r|\rho\|\sigma)= s_r-r.$ Denote by $T_n=\{\rho^{\otimes n}\le 2^{nr} \sigma^{\otimes n}\}$ the Neyman-Pearson test, i.e. the projector onto the non-negative subspace of $2^{nr} \sigma^{\otimes n}- \rho^{\otimes n}.$ Then
	\begin{align}
	\lim_{n\to\infty}\frac{-\log\Tr(T_n\rho^{\otimes n})}{n} = B(s_r|\rho\|\sigma), \quad \lim_{n\to\infty}\frac{-\log\Tr((\1-T_n)\sigma^{\otimes n})}{n} = s.
	\end{align}
	\end{lemma}

\section{Strong converse exponent in smoothed max-relative entropy}

For $\rho\in\cS(\cH)$, $\sigma\in\cP(\cH),$ $r\in\R$ and $d\in\{\cD,P\}$ we define by $\eps^{d}(\rho\|\sigma,r) $ the optimal, i.e. smallest exponent such that the corresponding smoothed max-relative entropy is less or equal to $r$, i.e. as 
	\begin{equation}
	\label{eq:epsDef}
	\eps^{d}(\rho\|\sigma,r) = \inf\Big\{d(\rho,\widetilde\rho)\Big|\, \widetilde\rho\le 2^{r}\sigma,\, \widetilde\rho\in\cS_{\le}(\cH)\Big\} = \inf\Big\{\eps\ge 0\Big|\, D_{\max}^{\eps,d}(\rho\|\sigma)\le r\Big\}.
	\end{equation}
    In their recent work \cite{LiYao_ExponentSmoothMaxEntr_2021}, Li, Yao and Hayashi proved that for $r> D(\rho\|\sigma)$ the optimal exponent in the i.i.d. setting, $\eps^{P}(\rho^{\otimes n}\|\sigma^{\otimes n}, nr)$ converges exponentially fast to $0$ as $n\to\infty.$ Moreover, they found that the corresponding exponential rate is given by
    \begin{equation}
    \lim_{n\to\infty}\frac{-\log\eps^{P}(\rho^{\otimes n}\|\sigma^{\otimes n}, nr)}{n} = \sup_{\alpha>1}\frac{\alpha-1}{2}\left(r- \widetilde D_\alpha(\rho\|\sigma)\right).
    \end{equation}

	Here, we are interested in the behaviour of the exponents $\eps^{d}(\rho^{\otimes n}\|\sigma^{\otimes n}, nr)$ in the region $r<D(\rho\|\sigma).$ We will see that in this case the exponents converge exponentially fast to 1 and by that establishing the corresponding strong converse property. Moreover, we provide a lower bound on the exponential rate of convergence.
	\begin{theorem} \label{thm:StrongConverse} 
For $\rho\in\cS(\cH)$, $\sigma\in\cP(\cH)$ and $r\in\R$ we have
		\begin{align}
		\label{eq:StrongConverse}
		\liminf_{n\to\infty}\frac{-\log(1-\eps^{\cD}(\rho^{\otimes n}\|\sigma^{\otimes n},nr))}{n} &\ge \sup_{0\le\alpha\le 1} (\alpha-1)\left(r - D_{\alpha}(\rho\|\sigma)\right),\\	\liminf_{n\to\infty}\frac{-\log(1-\eps^{P}(\rho^{\otimes n}\|\sigma^{\otimes n},nr))}{n} &\ge \sup_{0\le\alpha\le 1} (\alpha-1)\left(r - D_{\alpha}(\rho\|\sigma)\right).\label{eq:PuriExpStrong}
		\end{align}	
	Moreover, in the case in which $\rho$ and $\sigma$ commute, we have
	\begin{align}
	   \label{eq:StrongConverseEqual}
		\lim_{n\to\infty}\frac{-\log(1-\eps^{\cD}(\rho^{\otimes n}\|\sigma^{\otimes n},nr))}{n} &= \sup_{0\le\alpha\le 1} (\alpha-1)\left(r - D_{\alpha}(\rho\|\sigma)\right).
	\end{align}
	\end{theorem}
Note that indeed $\sup_{0\le\alpha\le 1} (\alpha-1)\left(r - D_{\alpha}(\rho\|\sigma)\right)$ is positive if and only if $r<D(\rho\|\sigma)$ which hence establishes the strong converse property of the exponent of the smoothed max-relative entropy. \\

\begin{proof}[Proof of Theorem~\ref{thm:StrongConverse}]
We first show the achievability bound \eqref{eq:StrongConverse}. Note that  \eqref{eq:PuriExpStrong}, which involves the purified distance, then follows immediately since using \eqref{eq:GenPuriBounds} we have the inequality
\begin{align*}
\eps^\cD(\rho\|\sigma,r) \le \eps^P(\rho\|\sigma,r).
\end{align*}
	Let now $\widetilde\rho_n\in\cS_{\le}(\cH^{\otimes n})$ be such that $\widetilde\rho_n\le2^{nr}\sigma^{\otimes n}$.
	Hence, by \eqref{eq:GenTrace}
	\begin{align}
	\label{eq:pERRbound}
	1- \cD(\rho^{\otimes n},\widetilde\rho_n) &=\nn \min_{0\le\Lambda_n\le\1}\left(1- |\Tr(\Lambda_n(\rho^{\otimes n}-\widetilde\rho_n))|\right) \\&\nn=  \min_{0\le\Lambda_n\le\1}\min\left\{1-\Tr(\Lambda_n(\rho^{\otimes n}-\widetilde\rho_n)),\,1-\Tr(\Lambda_n(\widetilde\rho_n-\rho^{\otimes n}))\right\}\\&=\nn \min_{0\le\Lambda_n\le\1}\left(\Tr\left(\rho^{\otimes n}(\1-\Lambda_n)\right) + \Tr\left(\widetilde\rho_n\Lambda_n\right) \right) \\&\le \min_{0\le\Lambda_n\le\1}\left(\Tr\left(\rho^{\otimes n}(\1-\Lambda_n)\right) + 2^{nr}\Tr\left(\sigma^{\otimes n}\Lambda_n\right) \right).
	\end{align}
	Now, by \eqref{eq:HoeffResult} there exists for all $s,\eps>0$ a sequence $0\le\Lambda_n(s,\eps)\le\1$ such that
	\begin{align*}
	\liminf_{n\to\infty}\frac{-\log(\Tr\left(\sigma^{\otimes n}\Lambda_n(s,\eps)\right))}{n} \ge s, \quad\quad	\liminf_{n\to\infty}\frac{-\log(\Tr\left(\rho^{\otimes}(\1-\Lambda_n(s,\eps))\right))}{n} \ge B(s|\rho\|\sigma) -\eps .
	\end{align*}
	Choosing $s= s_r$ where 
	\begin{align*}
	s_r := \sup_{0\le\alpha\le 1}\big(\alpha r -(\alpha-1)D_{\alpha}(\rho\|\sigma)\big)
	\end{align*}
	and noting that $s_r$ is the unique solution to the equation $B(s|\rho\|\sigma) = s - r$, we get for all $\eps>0$
	\begin{align*}
	\liminf_{n\to\infty}\frac{-\log(2^{nr}\Tr\left(\sigma^{\otimes n}\Lambda_n(s_r,\eps)\right))}{n} \ge s_r - r, \quad\quad	\liminf_{n\to\infty}\frac{-\log(\Tr\left(\rho^{\otimes}(\1-\Lambda_n(s_r,\eps))\right))}{n} \ge s_r - r -\eps.
	\end{align*}
	Therefore, using now \eqref{eq:pERRbound} we get
	\begin{align*}
	\liminf_{n\to\infty}\frac{-\log\left(1-\cD(\rho^{\otimes n},\widetilde\rho_n)\right)}{n} &\ge\liminf_{n\to\infty} \frac{-\log\big(\Tr\left(\rho^{\otimes n}(\1-\Lambda_n(s_r,\eps))\big) + 2^{nr}\Tr\left(\sigma^{\otimes n}\Lambda_n(s_r,\eps)\right) \right)}{n} \\&\ge s_r - r -\eps.
	\end{align*}
	Hence, by definition
	\begin{align*}
	\liminf_{n\to\infty}\frac{-\log(1-\eps^{\cD}(\rho^{\otimes n}\|\sigma^{\otimes n},nr))}{n} \ge s_r -r -\eps = \sup_{0\le\alpha\le 1} (\alpha-1)\left(r - D_{\alpha}(\rho\|\sigma)\right) - \eps,
	\end{align*}
	which finishes the proof of \eqref{eq:StrongConverse} since $\eps>0$ was arbitrary.

	We now consider $\rho$ and $\sigma$ to be commute and establish the converse in \eqref{eq:StrongConverseEqual}. For that we define the sub-normalised states $\widetilde\rho_n= T_n\rho^{\otimes n}T_n$, where $T_n = \{\rho^{\otimes n}\le 2^{nr}\sigma^{\otimes n}\}$ denotes the Neyman-Pearson test. Using that $\rho$ and $\sigma$ commute we see by definition  $\widetilde\rho_n\le 2^{nr}\sigma^{\otimes n}.$ Moreover, the generalised trace distance can be written as
	\begin{align*}
	\cD(\rho^{\otimes n},\widetilde\rho_n)& = \frac{1}{2}\left\|\rho^{\otimes n}- \widetilde\rho_n\right\|_1 + \frac{1}{2}\left|\Tr(\rho^{\otimes n} )- \Tr(\widetilde\rho_n)\right| \\&= \frac{1}{2} \|(\1-T_n)\rho^{\otimes n}(\1- T_n)\|_1  + \frac{1}{2}\left(1- \Tr(T_n\rho^{\otimes n})\right) \\&= 1- \Tr(T_n\rho^{\otimes n}).
	\end{align*}
	Here we have used for the first equality that since $\rho^{\otimes n}$ commutes with $\sigma^{\otimes n}$ and therefore also with $T_n$, the cross terms vanish, i.e. $T_n\rho^{\otimes n}(\1-T_n)= (\1-T_n)\rho^{\otimes n}T_n= 0.$
	Using now Lemma~\ref{lem:NeymanPearson} we see that 
	\begin{align*}
	\lim_{n\to\infty}\frac{-\log(1-\cD(\rho^{\otimes n},\widetilde\rho_n))}{n} = B(s_r|\rho\|\sigma) = \sup_{0\le\alpha\le 1}(\alpha-1)(r-D_{\alpha}(\rho\|\sigma)),
	\end{align*}
	and therefore
		\begin{align*}
	\limsup_{n\to\infty}\frac{-\log(1-\eps^{\cD}(\rho^{\otimes n}\|\sigma^{\otimes n}, nr)}{n} \le \sup_{0\le\alpha\le 1}(\alpha-1)(r-D_{\alpha}(\rho\|\sigma),
	\end{align*}
	which finishes the proof.

	\comment{For the converse we use the inequality 
	\begin{align*}
	D^{\eps}_{\max}(\rho\|\sigma) \le D^{1-\eps^2}_h(\rho\|\sigma) - \log(1-\eps^2)
	\end{align*}
	from \cite[Theorem 4]{AnshuBerta_MinMaxApproach_2019}. Hence,
	\begin{align}
	\label{eq:EpsForMaxandHypo}
	 \eps(\rho^{\otimes n}\|\sigma^{\otimes}| nr) = \inf\left\{\eps \Big| D_{\max}^\eps(\rho^{\otimes n}\|\sigma^{\otimes n}) \le nr\right\} \le \inf\left\{\eps \Big| D_{h}^{1-\eps^2}(\rho^{\otimes n}\|\sigma^{\otimes n}) - \log(1-\eps^2) \le nr\right\}.
	\end{align}
	Assume now \begin{align*}
	\limsup_{n\to\infty}\frac{ -\log(1- \eps(\rho^{\otimes n}\|\sigma^{\otimes}| nr) )}{n} \ge s \ge 0.
	\end{align*}
	By \eqref{eq:EpsForMaxandHypo} we therefore get
	\begin{align*}
	\sup_{\alpha}\frac{\alpha -1}{\alpha}\left(s-D_\alpha(\sigma\|\rho)\right) + s=B(s|\sigma\|\rho) + s = \liminf_{n\to\infty} \frac{D_h^{2^{-ns}}\left(\rho^{\otimes n}\|\sigma^{\otimes n}\right)}{n} + s \le r
	\end{align*}}
\end{proof}

\section{Strong converse rates of privacy amplification against a quantum adversary}
\label{sec:PrivacyAmplification}

The initial state of a privacy amplification protocol is given by the classical-quantum state
\begin{align}
\rho_{XE} = \sum_{x\in\X}p_x \kb{x}\otimes\rho_E^{x},
\end{align}
with the classical system $X$ belonging to Alice and Bob and the quantum system $E$ belonging to Eve.

 The objective of Alice and Bob is to apply a hash function $f:\X \to \Z$ to decouple their part of the system from Eve's system, the latter playing the role of quantum side information. Their operation results in the state
\begin{align}
\label{eq:HashedState}
    \rho^f_{ZE} =  \sum_{z\in \Z} \left( \kb{z} \otimes \sum_{x \in \cX:\atop{x \in f^{-1}(z)}} p_x \rho_E^x\right).
\end{align}
The minimal errors in this conversion, measured in trace distance and purified distance, respectively, are given by
\begin{align*}
\Delta_1(\X \to \Z) &= \min_{f}\frac{1}{2}\left\|\rho^{f}_{\Z E} - \frac{\1_\Z}{|\Z|}\otimes \rho_E\right\|_1, \\
\Delta_P(\X \to \Z) &= \min_{f}P\left( \rho^f_{ZE},\frac{\1_\Z}{|\Z|} \otimes \rho_E\right),
\end{align*}
where the minimum is over all functions $f:\X\to\Z.$ The \emph{optimal key length} which can be distilled from a c-q state $\rho_{XE}$ in that manner with $\eps\ge 0$ error measured in metric $d\in\{\cD,P\}$ is given by 
 
\begin{align*}
l^{\eps,d}(X|E) = \sup\left\{ \log|\Z| \,\Big|\,\Delta_d(\X\to\Z)\le \eps\right\}.
\end{align*}
Note that $l^{\eps,P}(X|E)$ is essentially given by the smoothed conditional min-entropy as we have for every $0<\eta\le \eps\le1$ the relation \cite[Theorem 8]{TomamichelHayashi_Hiearachy(2ndOrderExpansions)_2014} 
\begin{equation}
\label{eq:KeyLengthMinEnt}
H^{\eps,P}_{\min}(X|E) \ge l^{\eps,P}(X|E) \ge H^{\eps-\eta,P}_{\min}(X|E) - \log\left(1/\eta^4\right)- 3. \footnote{Note that in \cite{TomamichelHayashi_Hiearachy(2ndOrderExpansions)_2014} the result is proven for a different version of the smoothed conditional min-entropy, namely, $H_{\min}^{\eps,\uparrow}(X|E) := \sup_{\sigma_{E}\in\cS(\cH_E)} - D^{\eps,P}_{\max}(\rho_{XE}\|\1_X\otimes\sigma_E).$ Since clearly $H_{\min}^{\eps,\uparrow}(X|E) \ge H^{\eps,P}_{\min}(A|B)$, their lower bound on $l^{\eps,P}(X|E)$ implies the one in \eqref{eq:KeyLengthMinEnt}. Moreover, using the result \cite[Proposition 10]{LiYao_ExponentSmoothMaxEntr_2021} and the same proof as in~\cite[Theorem 8]{TomamichelHayashi_Hiearachy(2ndOrderExpansions)_2014}, the upper bound on $l^{\eps,P}(X|E)$ in \eqref{eq:KeyLengthMinEnt} also follows.}
\end{equation}

Let us go to the $n$-copy setting in which Alice and Bob possess strings $x^{(n)}=(x_1, \ldots, x_n) \in \cX^n$ with the $x_i$'s being values taken by a sequence of i.i.d.~random variables with common p.m.f.~$p_x$, $x \in \cX$. \comment{The strings are of course not completely private due to the quantum side information that Eve possesses.} The initial state of the privacy amplification process is in this case:
\begin{align}
    \rho_{X^nE^n} := \sum_{x^{(n)} \in \cX^n} p_{x^{(n)}} \ket{x^{(n)}}\bra{x^{(n)}}\otimes  \rho_{E^n}^{x^{(n)}},
\end{align}
and the desired final state is $\frac{\1_{\Z_n}}{|\Z_n|} \otimes \rho_E^{\otimes n}$.

For a  privacy amplification rate $R\ge0$ we consider $\Z_n$ to be such that $|\Z_n| = \left \lfloor{2^{nR}}\right \rfloor$. 
\comment{$\cF_n(R)$ denote a 2-universal family of hash functions from $\cX_n$ to $\Z_n$ of {\em{extraction rate}} $R$:
\begin{align}
    \cF_n(R) :=\left\{ f_n:  \cX_n \to \Z_n \,|\, | \Z_n|= 2^{nR} \,P \left(F_n(x^{(n)}) = F_n({\tilde{x}}^{(n)})\right) \leq \frac{1}{|\Z_n|}, \,\forall\, {{x}}^{(n)}, {\tilde{x}}^{(n)} \in \cX^n\right\},
\end{align}
where $F_n$ denotes a random choice (selected uniformly at random) of $f_n \in \cF_n(R)$.
Alice and Bob apply a 2-universal hash function from $\cF_n(R)$ to the strings in their possession. The final state is given by:
\begin{align}\label{final}
    \rho^{f_n}_{Z_nE^n} := \sum_{z^{(n)}\in \Z_n} \left( \ket{z^{(n)}}\bra{z^{(n)}} \otimes \sum_{x^{(n)} \in \cX^n:\atop{x^{(n)} \in f_n^{-1}(z^{(n)})}} p_{x^{(n)}} \rho_E^{x^{(n)}}\right).
\end{align}}
The next theorem shows that if the randomness extraction rate in a transformation from $\X^n$ to $\Z_n$ is larger than the conditional entropy $H(X|E)$, the distance of the corresponding decoupled state to any possible $\rho^{f_n}_{Z_nE^n}$, defined analogously to \eqref{eq:HashedState}, goes exponentially fast to 1 as $n$ goes to infinity. By that it provides the strong converse of privacy amplification. Moreover, it provides bounds on the corresponding strong converse rates.

\begin{theorem}
\label{thm:PrivStrong}
Let $\rho_{XE}$ be a c-q state, $R\ge 0$ and $\Z_n$ with $|\Z_n| = \left \lfloor{2^{nR}}\right \rfloor.$ Then we have
\begin{align}
\label{eq:PrivStrongTrace}\liminf_{n\to\infty} \frac{-\log\left(1- \Delta_1(\X^n\to\Z_n)\right)}{n} &\ge \sup_{0 \le \alpha \le 1} \frac{(1- \alpha )}{2} (R-H_\alpha(X|E)), \\ \label{eq:PrivPuri}
\liminf_{n\to\infty} \frac{-\log\left(1- \Delta_P(\X^n\to\Z_n)\right)}{n} &\ge \sup_{0 \le \alpha \le 1} (1- \alpha )  (R-H_\alpha(X|E)).
\end{align}

\end{theorem}
\begin{remark}
Note that the right hand side of both \eqref{eq:PrivPuri} and \eqref{eq:PrivStrongTrace} are strictly positive if only if $R> H(X|E)$. This in turn proves the full strong converse, since $H(X|E)$ is the optimal achievable extraction rate of privacy amplification against a quantum adversary (c.f. \eqref{eq:AEPkeyrate}). In other words, the optimal extraction rate of privacy amplification against a quantum adversary satisfies the so-called {\em{strong converse property}}: i.e., an attempt to extract private bits at a rate higher than $H(X|E)$ leads to the conversion error going to 1 exponentially fast as $n \to \infty.$
\end{remark}
\begin{remark}
Note that in the concurrent and independent work \cite{ShenGaoCheng_StrongConvPrivAmpl_2022}, the lower bound
\begin{align}
\label{eq:HaoChungBound}
\frac{-\log\left(1- \overline\Delta_1(\X^n\to\Z_n)\right)}{n} &\ge \sup_{0 \le \alpha \le 1} (1- \alpha ) (R-H_\alpha(X|E)),
\end{align}
has been established. Here, $\overline\Delta_1(\X^n\to\Z_n) = \bE_{h_n} \, \frac{1}{2}\|\rho^{h_n}_{Z_nE^n} - \frac{\1_{\Z_n}}{|\Z_n|}\otimes\rho^{\otimes n}_{E}\|_1$ denotes the average trace distance from the decoupled state, with $\bE_{h_n}$ being the average over all strongly 2-universal hash functions (see \cite{ShenGaoCheng_StrongConvPrivAmpl_2022} for more details). Note that due to the factor of $1/2$ the right-hand side of our \eqref{eq:PrivStrongTrace} is smaller than the right-hand side of \eqref{eq:HaoChungBound}. This factor arises solely from the fact that we use our result for purified distance, \eqref{eq:PrivPuri}, together with the Fuchs-van de Graaf inequality \cite{FuchsvandeGraaf_Inequality_1999}. If the following conjectured inequality \eqref{eq:LisProp} holds:
\begin{align*}
H^{\eps,\cD}_{\text{min}}(X^n|E^n)_{\rho_{X^nE^n}} \stackrel{?}{\ge} H^{\eps,\cD}_{\text{min}}(Z_n|E^n)_{\rho^{f_n}_{Z_nE^n}},
\end{align*}
which is the analogue of \eqref{eq:LisProp} but with smoothing considered in terms of the generalised trace distance instead of  purified distance, we can obtain a lower bound on the strong converse exponent \eqref{eq:PrivStrongTrace} with an improvement of a factor of $2.$ The resulting bound would then match the lower bound \eqref{eq:HaoChungBound} from  \cite{ShenGaoCheng_StrongConvPrivAmpl_2022}. In fact, since we consider the {\em{minimal}} conversion error, $\Delta_1(\X^n\to\Z_n)$, rather than the {\em{average}} one $\overline\Delta_1(\X^n\to\Z_n)$, this would give a stronger result on the strong converse exponent of privacy amplification.
\end{remark}
\begin{remark}
In~\cite{LeditzkyWildeDatta_StrongConverseRenyiEnt_2016}, the authors obtained alternative bounds on the strong converse exponent for quantum privacy amplification, using a figure of merit which is different from the ones considered in this paper. The figure of merit in~\cite{LeditzkyWildeDatta_StrongConverseRenyiEnt_2016} is the fidelity (between the final c-q state of privacy amplification protocol and a decoupled state corresponding to uniform shared randomness of the classical system and any possible quantum state of Eve) optimized over all possible states of Eve.

\end{remark}
\comment{\begin{proof}[Proof of Theorem~\ref{thm:PrivStrong}]
We define 
\begin{align*}
\eps(\rho_{XA}, r) = \inf\left\{\eps\ge 0\Big|\, l^{\eps}(X|E) \ge r\right\}.
\end{align*}
Clearly, as for every $f:\cX \to \Z$ and $\rho^{f}_{ZE} = \cN_f(\rho_{XE})$ we have
\begin{align*}
H_{\min}^{\eps}(X|E)_\rho \ge l^{\eps}(X|E)_{\rho} \ge l^{\eps}(Z|E)_{\rho^f} 
\end{align*}
and hence for every $r\in\R$
\begin{align*}
\eps(\rho^f_{ZE}, r) \ge \eps(\rho_{XE}, r) \ge \eps(\rho_{XA},\1_{\X}\otimes\rho_{E},-r)
\end{align*}
Moreover, we have
\begin{align*}
\frac{1}{2}\left\|\rho^f_{ZE} - \frac{\1_\Z}{|\Z|}\otimes\rho_E\right\|_1  \ge \eps(\rho^f_{ZE}, \log|\Z|).
\end{align*}
\end{proof}}

\begin{proof}[Proof of Theorem~\ref{thm:PrivStrong}]
We first prove \eqref{eq:PrivPuri} which essentially follows by the same method as in \cite[Theorem 8]{LiYao_ExponentSmoothMaxEntr_2021}, but we still carry it out for completeness. From \cite[Proposition 10]{LiYao_ExponentSmoothMaxEntr_2021}
we know for all $\eps\ge 0$ and $f_n:\X^n \to \Z_n$
\begin{align}
\label{eq:LisProp}
H^{\eps,P}_{\text{min}}(X^n|E^n)_{\rho_{X^nE^n}} \ge H^{\eps,P}_{\text{min}}(Z_n|E^n)_{\rho^{f_n}_{Z_nE^n}}.
\end{align}
This gives for all $r\in\R$
\begin{align}
\label{eq:Epsineq}
\eps^P\left(\rho^{f_n}_{Z_nE^n}\big\|\1_{\Z_n}\otimes\rho_{E^n},r\right) \ge \eps^P\left(\rho_{X^nE^n}\big\|\1_{\X_n}\otimes\rho_{E^n},r\right).
\end{align}
By definition we have
\begin{align*}
P\left( \rho^{f_n}_{\Z_nE^n}, \frac{\1_{\Z_n}}{|\Z_n|} \otimes \rho_E^{n}\right) \geq \eps^P\left( \rho^{f_n}_{Z_nE^n}\Big\| {\1_{\Z_n}}\otimes \rho_E, - \log |\Z_n|\right) \ge \eps^P\left( \rho^{f_n}_{\Z_nE^n}\Big\| {\1_{\Z_n}}\otimes \rho_E, -nR\right),
\end{align*}
where we have used $\log|\Z_n| = \log\left\lfloor {2^{nR}}\right\rfloor \le nR$ for the last inequality.
Combining this with \eqref{eq:Epsineq} and using that $f_n$ was arbitrary gives
\begin{align*}
\Delta_P(\X^n\to\Z_n) &\ge \eps^P\left(\rho_{X^nE^n}\Big\|\1_{\X_n}\otimes\rho_{E^n}, -nR\right).
\end{align*}
Using now Theorem~\ref{thm:StrongConverse} gives
\begin{align*}
\liminf_{n\to\infty}\frac{-\log\left(1 - \Delta_P(\X^n\to\Z_n)\right)}{n} &\ge \sup_{0\le\alpha\le 1}(1-\alpha)\left(R + D_\alpha\left(\rho_{XE}\|\1_{\X}\otimes\rho_E\right)\right) \\&=\sup_{0\le\alpha\le 1}(1-\alpha)\left(R - H_\alpha\left(X \big|E\right)\right).
\end{align*}
For \eqref{eq:PrivStrongTrace} we use that by the Fuchs-van de Graaf inequality \cite{FuchsvandeGraaf_Inequality_1999} we have for all functions $f_n:\X^n\to\Z_n$ 
\begin{align*}
1- \frac{1}{2}\left\|\rho^{f_n}_{Z_nE^n} -  \frac{\1_{Z_n}}{|\Z_n|} \otimes \rho_E^{n}\right\|_1 &\le \sqrt{1- P^2\left(\rho^{f_n}_{Z_nE^n}, \frac{\1_{Z_n}}{|\Z_n|} \otimes \rho_E^{n}\right)} \\&= \sqrt{\left(1+ P\left(\rho^{f_n}_{Z_nE^n}, \frac{\1_{Z_n}}{|\Z_n|} \otimes \rho_E^{n}\right)\right)\left(1- P\left(\rho^{f_n}_{Z_nE^n}, \frac{\1_{\Z_n}}{|\Z_n|} \otimes \rho_E^{n}\right)\right)} \\& \le \sqrt{2\left(1- P\left(\rho^{f_n}_{Z_nE^n}, \frac{\1_{\Z_n}}{|\Z_n|} \otimes \rho_E^{n}\right)\right)}
\end{align*}
and therefore by \eqref{eq:PrivPuri}
\begin{align*}
&\liminf_{n\to\infty} \frac{-\log\left(1- \Delta_1(\X^n\to\Z_n)\right)}{n}\ge \frac{1}{2}\liminf_{n\to\infty} \frac{-\log\left(1- \Delta_P(\X^n\to\Z_n)\right)}{n} \\&\ge\sup_{0 \le \alpha \le 1} \frac{(1- \alpha )}{2}  (R-H_\alpha(X|E)_\rho)
\end{align*}
which finishes the proof.

\end{proof}

\section{Strong converse for secure communication}
In the following we consider $Z_n$  to be a sequence of classical systems, $E_n$ a sequence of quantum systems. Moreover, we consider sequences of c-q states denoted by
\begin{align}
    \rho_{Z_nE_n} = \sum_{z_n\in\Z_n} p_{z_n} \kb{z_n}\otimes \rho^{z_n}_{E_n},
\end{align}
where we assume the system $Z_n$ to be held by Alice and $E_n$ to be held by an evesdropper, Eve.
we write $\rho_{E_n} = \Tr_{Z_n}(\rho_{Z_nE_n} )= \sum_{z_n\in\Z_n} p_{z_n} \rho^{z_n}_{E_n}$ for the corresponding reduced states of Eve's system. 

Here, we want to understand in what sense the condition
\begin{align}
\label{eq:AntiDecoupling}
    \lim_{n\to\infty}\frac{1}{2}\left\|\rho_{Z_nE_n} - \frac{\1_{\Z_n}}{|\Z_n|}\otimes \rho_{E_n}\right\|_1=1
\end{align}  implies insecurity of the keys Alice can generate from $\rho_{Z_nE_n}$. The example we have in mind is $\rho_{Z_nE_n}$ being the resulting state after privacy amplification  with privacy amplification rate $R>H(X|E)$, where we have already seen that the convergence in \eqref{eq:AntiDecoupling} happens exponentially fast (c.f. Theorem~\ref{thm:PrivStrong}). 

We consider the scenario in which Alice is using a generated key $z_n\in\Z_n$ to encode a message $m$, where the message is taken out of a subset $\cM_n\subset\Z_n.$ For the encoding she uses an \emph{encryption scheme}, i.e. a map
\begin{align*}
\cE : \Z_n\times\Z_n &\to \Z_n \\
            (z_n,m) &\mapsto \cE_{z_n}(m),
\end{align*}
which is bijective in both entries for the other entry fixed. The encoded message $\cE_{z_n}(m)$ is then sent publicly. Given a party has access to the key $z_n$, they can then decode $m$ due to bijectivity of the function $\cE_{z_n}.$ As an example of such an encryption scheme we can think of $\Z_n$ being a set of bit strings and take the one-time pad encoding $\cE_{z_n}(m) = z_n\oplus m,$ where $\oplus$ denotes the component-wise addition modulo 2.

Since the encoded message is sent publicly, Eve has access to the c-q state 
\begin{align*}
\rho^{m}_{\cE_{Z_n}(M) E_n} = \sum_{Z_n\in\Z_n} p_{z_n} \,\kb{\cE_{z_n}(m)}\otimes\rho^{Z_n}_{E_n} = (U_m\otimes\1_{E_n}) \rho_{Z_nE_n}(U^*_m\otimes\1_{E_n}).
\end{align*}
Here, we have denoted the unitary $U_m$ on system $Z_n$ defined by $U_m\ket{z_n} = \ket{\cE_{z_n}(m)}.$ In order to infer which message $m\in\cM_n\subset \Z_n$ has been sent, Eve needs to distinguish the states $\left(\rho^{m}_{\cE_{Z_n}(M) E_n}\right)_{m\in\cM_n}$ by picking a suitable POVM $\Lambda\equiv\left(\Lambda_m\right)_{m\in\cM_n}.$ Given that Alice choses the message $M=m$, the probability that Eve's guess, denoted by $\hat M$, is correct is given by
\begin{align*}
\bP_{\Lambda,n}(\hat M = m| M = m) = \Tr(\Lambda_m\rho^{m}_{\cE_{Z_n}(M) E_n}).
\end{align*}
Moreover, if Alice message is distributed by some distribution $M\sim q_m$ on $\cM_n$ then the optimal probability with which Eve guesses correctly on average, the \emph{guessing probability} in the following, is given by
\begin{align}
\label{eq:GuessProb}
\nn p_{\text{guess}}(M|\cE_{Z_n}(M),E_n) &=\max_{\substack{ (\Lambda_m)_{m\in\cM_n}\\\text{ POVM}}}\sum_{m\in\cM_n} q_m \bP_{\Lambda,n}(\hat M = m| M = m)\\&=  \max_{\substack{ (\Lambda_m)_{m\in\cM_n}\\\text{ POVM}}}\sum_{m\in\cM_n} q_m\Tr\left(\Lambda_m \rho^{m}_{\cE_{Z_n}(M) E_n}\right).
\end{align}

However, in order for Eve to be able to pick the POVM sensefully, i.e. the maximiser in \eqref{eq:GuessProb}, she needs to have additonal side information, which is knowledge of the set $\cM_n$ and the distribution $q_m.$ Here, we say her additional side information is large if $\cM_n$ is small. In the following we will quantify how much additional side information Eve needs to have in order to guess Alice's message with certainty.

We first consider the case in which \eqref{eq:AntiDecoupling} does not hold, i.e. the trace distance between $\rho_{Z_nE_n}$ and the decoupled state is strictly smaller than 1 uniformly in $n.$ In that case, it is only possible for Eve to have certainty of the sent message if she has strong additional side information which is that $\cM_n \subset \Z_n$  is finite uniformly in $n$, i.e. $\sup_{n\in\N}|\cM_n|<\infty.$ 
\begin{proposition}
\label{prop:FiniteConvError}
Let $\delta \in[0,1)$ and $\left(\rho_{Z_nE_n}\right)_{n\in\N}$ be a sequence of c-q states such that
\begin{align}
\label{eq:NonStrongConverse}
\sup_{n\in\N} \frac{1}{2}\left\|\rho_{Z_nE_n} - \frac{\1_{\Z_n}}{|\Z_n|} \otimes \rho_{E_n}\right\|_1 \le \delta.
\end{align}
Moreover, assume Alice uses an encryption scheme $g$ as defined above and chooses to send a message $M$ which she picks uniformly out of the set $\cM_n\subset\Z_n$, i.e. $q_m = 1/|\cM_n|$ for all $m\in\cM_n.$ Then Eve's guessing probability is bounded by
\begin{align*}
p_{\text{guess}}(M|\cE_{Z_n}(M), E_n)  \le \delta + \frac{1}{|\cM_n|}.
\end{align*}
In particular, if $\sup_{n\in\N}|\cM_n|=\infty$ then 
\begin{align*}
\limsup_{n\to\infty}p_{\text{guess}}(M|\cE_{Z_n}(M), E_n)  \le \delta < 1.
\end{align*}

\end{proposition}

\begin{proof}
\comment{Note that Eve's guessing probability is given by
\begin{align*}
p_{\text{guess}}(M|\cE_{Z_n}(M), E_n)  = \sup_{\substack{\{\Lambda_m\}_{m\in\cM_n}\\\text{ POVM}}}\frac{1}{|\cM_n|}\sum_{m\in\cM_n} \Tr\left(\Lambda_m \rho^{m}_{\cE_{Z_n}(M) E_n}\right),
\end{align*}
where $\cE_{z_n}(m) E_n$ is the quantum system hold by Eve containing the public encrypted message and the side information in $E_n$ with state
\begin{align*}
\rho^{m}_{\cE_{Z_n}(M) E_n} = \sum_{Z_n\in\Z_n} p_{Z_n} \,\kb{\cE_{z_n}(m)}\otimes\rho^{Z_n}_{E_n} = U_m \rho_{Z_nE_n}U^*_m.
\end{align*}
Here we denoted the unitary $U_m$ which is defined by the action $U_m\ket{Z_n} = \ket{\cE_{z_n}(m)}.$}
From \eqref{eq:NonStrongConverse} and unitary invariance of the trace norm we get for all $m\in\Z_n$

\begin{align*}
\label{eq:NonStrongConverse}
 \sup_{n\in\N}\frac{1}{2}\left\|\rho^{m}_{\cE_{Z_n}(M)E_n} - \frac{\1_{\Z_n}}{|\Z_n|}\otimes\rho_{E_n}\right\|_1 =  \sup_{n\in\N}\frac{1}{2}\left\|\rho_{Z_nE_n} - \frac{\1_{\Z_n}}{|\Z_n|} \otimes \rho_{E_n}\right\|_1 \le \delta.
\end{align*}
Using this we get
\begin{align*}
p_{\text{guess}}(M|\cE_{Z_n}(M), E_n)   &= \max_{\substack{ (\Lambda_m)_{m\in\cM_n}\\\text{ POVM}}}\frac{1}{|\cM_n|}\sum_{m\in\cM_n} \Tr\left(\Lambda_m \rho^{m}_{\cE_{Z_n}(M) E_n}\right) \\ &\le \max_{\substack{ (\Lambda_m)_{m\in\cM_n}\\\text{ POVM}}}\frac{1}{|\cM_n|}\sum_{m\in\cM_n} \Tr\left(\Lambda_m\, \frac{\1_{\Z_n}}{|\Z_n|} \otimes \rho_{E_n}\right) + \delta  \\ &=\frac{1}{|\cM_n|} +\delta.
\end{align*}
\end{proof}

In the case where \eqref{eq:AntiDecoupling} does hold, the situation is drastically different to the one described in Proposition~\ref{prop:FiniteConvError}. Let  $c_n\in[0,1]$ be the speed of convergence in \eqref{eq:AntiDecoupling},  i.e. 
\begin{align}
\left| 1-\frac{1}{2}\left\|\rho_{Z_n E_n} - \frac{\1_{\Z_n}}{|\Z_n|} \otimes \rho_{E_n}\right\|_1\right| \le c_n
\end{align}
and $\lim_{n\to\infty} c_n=0.$ We will show that for all encryption schemes and almost all sets of messages with $|\cM_n| \ll c_n^{-1}$, Eve's guessing probability will converge to 1 as $n\to\infty$ (c.f. Proposition~\ref{prop:StrongConverseMessages}). In particular, in the strong converse region of privacy amplification, i.e. for privacy amplification rate $R>H(X|E)$, Theorem~\ref{thm:PrivStrong} gives that $c_n$ converges exponentially fast to 0 with exponential decay rate being bounded from below as
\begin{align*}
 \liminf_{n\to\infty}\frac{-\log c_n}{n} \ge \sup_{0\le\alpha\le1} \frac{(1-\alpha)}{2}\left( R -H_\alpha(X|E)\right).
\end{align*}
Hence, the sets of messages for which Eve can guess correctly can be made even exponentially large in $n$, which can be considered as small additional side information needed. Moreover, it can be shown that the convergence of Eve's guessing probability towards 1 happens exponentially fast in that region.

Therfore, to conclude, we see by Proposition~\ref{prop:FiniteConvError} that in the achievable region of privacy amplification for privacy amplification rate $R\le H(X|E)$ in which
\begin{align*}
\lim_{n\to\infty}\Delta_1(\X^n\to\Z_n) = 0,
\end{align*}
(compare \cite{Renner_PhDThesis_2005})
Eve can infer the sent message correctly only if she already had complete knowledge of it to start with (i.e. $|\cM_n|=1$). Whereas, the following Proposition~\ref{prop:StrongConverseMessages} gives in the strong converse region for $R> H(X|E)$, Eve only needs very limited additional side information ($|\cM_n|$ possibly scaling exponentially in $n$) while still being able to infer the message correctly. Hence, our result serves as a strong converse for secure communication.

\comment{\begin{remark}

In the strong converse region of privacy amplification considered in Section~\ref{sec:PrivacyAmplification}, the convergence in (26) is exponential in $n$, i.e. 
$c_n \sim 2^{-n sc(R)}$ asymptotically. 
Therefore, Proposition~\ref{prop:StrongConverseMessages} gives that we can choose an exponentially scaling set of messages $\cM_n$, i.e. $$|\cM_n| \sim 2^{n(1-\eps)sc(R)}$$ such that Eve's guessing probability goes exponentially fast to 1.
\end{remark}}

\begin{proposition}
\label{prop:StrongConverseMessages}
Let $\left(\rho_{Z_nE_n}\right)_{n\in\N}$ be a sequence of c-q states such that
\begin{equation}
\label{eq:StrongTrace}
\lim_{n \to \infty}
\frac{1}{2}\left\|\rho_{Z_nE_n} - \frac{\1_{\Z_n}}{|\Z_n|} \otimes \rho_{E_n}\right\|_1 = 1,
\end{equation}
with speed of convergence $c_n\in[0,1]$,  i.e. 
\begin{align}
\label{eq:StrongConverse1}
\left| 1-\frac{1}{2}\left\|\rho_{Z_nE_n} - \frac{\1_{\Z_n}}{|\Z_n|} \otimes \rho_{E_n}\right\|_1\right| \le c_n
\end{align}
and $\lim_{n\to\infty}c_n =0.$ Let $\eps>0$ and $\cM_n\subset \Z_n$ be a set of messages with $|\cM_n| \le c_n^{(\eps-1)}$ chosen uniformly at random, which Alice encodes using an encryption scheme $\cE$ as above. Then Eve can find a measurement $\Lambda$ such that for all $m\in\cM_n$
\begin{equation}
\bP_{\Lambda,n}(\hat M=m| M=m) > 1- c_n^{\eps/2}
\end{equation}
with probability greater than $1-2c_n^{\eps/2}.$
Hence, in particular for all $m\in\cM_n$ we have $\lim_{n\to\infty}\bP_{\Lambda,n}(\hat M=m|M=m) =1$
almost surely.
\end{proposition}
\begin{remark}
Note that by Proposition~\ref{prop:StrongConverseMessages} we immediately also get the lower bound on the average guessing probability (for all possible distributions $q_m$ on $\cM_n$)
\begin{align*}
p_{\text{guess}}(M|\cE_{Z_n}(M),E_n) >   1- c_n^{\eps/2}
\end{align*}
with probability greater than $1-2c_n^{\eps/2}$ over the set $\cM_n$ chosen uniformly at random from the set of subsets of $\Z_n$ with cardinality constraint $|\cM_n|\le c_n^{(\eps-1)}.$
\end{remark}
\begin{proof}[Proof of Proposition~\ref{prop:StrongConverseMessages}]
Using the well-known expression for the trace distance of two states $\rho$ and $\sigma$
\begin{align*}
\frac{1}{2}\|\rho-\sigma\|_1 = \max_{0\le\Lambda\le \1}\Tr\left(\Lambda(\rho-\sigma)\right) =  \max_{\pi \text{ orthogonal projection}}\Tr\left(\pi(\rho-\sigma)\right),
\end{align*}
assumption \eqref{eq:StrongConverse1} gives that there exists a sequence of orthogonal projections $\left(\pi_n \right)_{n\in \N}$ such that
\begin{align*}
\Tr\left(\pi_n \rho_{Z_n E_n}\right) \ge 1- c_n,\\
\Tr\left(\pi_n \frac{\1_{Z_n}}{|\Z_n|}\otimes \rho_{E_n}\right) \le c_n.
\end{align*}
Moreover, define as above for $m\in\Z_n$ 
\begin{align*}
\rho^{m}_{\cE_{Z_n}(M)E_n} = \sum_{z_n\in\Z_n}p_{z_n} \kb{z_n\oplus m}\otimes \rho^{z_n}_{E_n} = (U_m\otimes \1_{E_n})\rho_{Z_nE_n}(U_m^*\otimes\1_{E_n}),
\end{align*}
with unitary $U_m$ on system $Z_n$ defined by $U_m\ket{x} = \ket{\cE_{z_n}(m)}.$  Furthermore, we write $\pi_n^m = (U_m^*\otimes\1_{E_n})\pi_n (U_m\otimes\1_{E_n}).$ 
Note
\begin{align*}
\frac{1}{|\Z_n|}\sum_{m\in\Z_n}\rho^{m}_{\cE_{Z_n}(M)E_n} = \sum_{z_n\in\Z_n} p_{z_n} \frac{\1_{\Z_n}}{|\Z_n|}\otimes \rho^{z_n}_{E_n} = \frac{\1_{\Z_n}}{|\Z_n|}\otimes \rho_{E_n},
\end{align*}
where we have used for the first equality that for every $z_n\in\Z_n$ the function $\cE_{z_n}$ is a bijection.
For $K_n\le \floor{1/c_n^{(1-\eps)}}$ being a natural number let $\cM_n = \left(m_1,\cdots,m_{K_n}\right) \in\Z_n^{K_n}$ be random vector  such that each component of $\cM_n$ is picked uniformly at random from $\Z_n.$ For every $k \in[K_n]:=\{1,\cdots,K_n\}$ we can calculate the expectation
\begin{align*}
&\bE_{\cM_n}\left[\Tr\left(\rho^{m_k}_{\cE_{Z_n}(M) E_n}\sum_{l\in[K_n]\setminus\{k\}}\pi^{m_l}_n\right)\right] = \frac{1}{|\Z_n|^{K_n}}\sum_{m_1,\cdots,m_{K_n} \in\Z_n}\Tr\left(\rho^{m_k}_{\cE_{Z_n}(M) E_n}\sum_{l\in[K_n]\setminus\{k\}}\pi^{m_l}_n\right) \\
&=\frac{1}{|\Z_n|^{K_n-1}}\sum_{m_1,\cdots,m_{k-1},m_{k+1},\cdots,m_K \in\Z_n}\Tr\left( \frac{\1_{\Z_n}}{|\Z_n|}\otimes \rho_{E_n}\sum_{l\in[K_n]\setminus\{k\}}\pi^{m_l}_n\right) \\&=\sum_{l\in[K_n]\setminus\{k\}}\Tr\left( \frac{\1_{\Z_n}}{|\Z_n|}\otimes \rho_{E_n}\pi_n\right) \le |\cM_n| c_n \le c_n^\eps
\end{align*}

Define for all $m\in\cM_n$ the subspaces of $\cH_{\Z_n}\otimes\cH_{E_n}$ denoted by $V_m = \supp(\pi_n^m)\cap\supp(\sum_{m'\in\cM_n\setminus\{m\}}\pi_n^{m'})^\perp$ and $\Lambda_n^m$ the orthogonal projection on $V_m.$ Note, that by definition all $\Lambda^m_n$ have mutually orthogonal supports and hence $\sum_{m\in\cM_n}\Lambda^m_n \le \1$, which means we can extend the family $(\Lambda_n^m)_{m\in\cM_n}$ to a POVM (i.e. by redefining $\Lambda^{m}_n = \1 -\sum_{m'\in\cM_n\setminus\{m\}}\Lambda^{m'}_n$ for one fixed $m\in\cM_n$).

 Using the notation $P_V$ for the orthogonal projection onto a subspace $V\subset \cH_{\Z_n}\otimes \cH_{E_n}$ we note for all $m\in\cM_n$
\begin{align*}
\bE_{\cM_n}\left[\bP_{\Lambda,n}(\hat M=m | M=m)\right]&= 
\bE_{\cM_n}\left[\Tr\left(\Lambda_n^{m}\rho^{m}_{\cE_{Z_n}(M) E_n}\right)\right] \\&= \bE_{\cM_n}\left[\Tr\left(\pi_n^{m}\rho^{m}_{\cE_{Z_n}(M) E_n}\right) -  \Tr\left(P_{\supp(\pi^m_n)\cap\supp(\sum_{m'\in\cM_n\setminus\{m\}}\pi_n^{m'})}\,\rho^{m}_{\cE_{Z_n}(M) E_n}\right)\right] \\&\ge \Tr\left(\pi_n\rho_{Z_n E_n}\right) -  \bE_{\cM_n}\left[\Tr\left(P_{\supp(\sum_{m'\in\cM_n\setminus\{m\}}\pi _n^{m'})}\,\rho^{m}_{\cE_{Z_n}(M) E_n}\right)\right]\\&\ge \Tr\left(\pi_n^{m}\rho^{m}_{\cE_{Z_n}(M) E_n}\right) - \bE_{\cM_n}\left[\Tr\left(\rho^{m}_{\cE_{Z_n}(M) E_n}\sum_{\substack{m'\in\cM_n\\m'\neq m}} \pi_n^{m'}\right)\right] \\&\ge 1- c_n -c_n^{\eps} \ge 1 -2 c_n^\eps.
\end{align*}
Moreover, by Markov's inequality, we know that for every $\delta>0$ 
\begin{align*}
\bP_{\cM_n}\big[\,\bP_{\Lambda,n}(\hat M=m|M=m)\le 1-\delta \big]\le \frac{\bE_{\cM_n}\left[1-\bP_{\Lambda,n}(\hat M=m|M=m)\right]}{\delta} \le \frac{2c_n^\eps}{\delta}
\end{align*}
Hence, choosing $\delta = c_n^{\eps/2}$ gives the desired result.
\end{proof}

\bigskip

\noindent\textbf{Acknowledgments.} The authors would like to thank Bjarne Bergh, Renato Renner, Marco Tomamichel and Mark M. Wilde for helpful comments.
RS gratefully acknowledges support from the Cambridge Commonwealth, European and International Trust.

\bibliography{Ref}
\bibliographystyle{abbrv}

\end{document}